\documentclass[aps,pra,reprint,notitlepage,nobalancelastpage]{revtex4-1}

\usepackage[utf8]{inputenc} 
\usepackage{graphicx}  

\usepackage{amsmath,amsthm}
\usepackage{bm}  
\usepackage{bbm}

\usepackage[bitstream-charter]{mathdesign}
\usepackage[T1]{fontenc}

\usepackage[margin=2cm,a4paper]{geometry}
\usepackage{verbatim}
\usepackage{xcolor}
\usepackage{units}

\usepackage[pdftex,bookmarks,colorlinks,breaklinks]{hyperref}  
\definecolor{dullmagenta}{rgb}{0.4,0,0.4}   
\definecolor{darkblue}{rgb}{0,0,0.4}
\hypersetup{linkcolor=red,citecolor=blue,filecolor=dullmagenta,urlcolor=darkblue} 
\usepackage{pdfsync}  

\usepackage{tikz}
\usepackage{pgfplots}
\pgfplotsset{compat=1.3}
\pgfplotsset{width=0.5*\textwidth}

\usepackage{braket}

\def\tr{{\rm Tr}}

\def\eps{\epsilon}

\usepackage{xspace}
\def\etal{\emph{et al.\xspace}}

\newtheorem{theorem}{Theorem}

\newtheorem{prop}{Proposition}

\begin{document}
\title{Work Cost of Thermal Operations in Quantum and Nano Thermodynamics}
\author{Joseph M. Renes}
\affiliation{Institute for Theoretical Physics, ETH Zurich, 8093 Zurich, Switzerland}

\begin{abstract}
Adopting a resource theory framework of thermodynamics for quantum and nano systems pioneered by Janzing \etal\ [Int.\ J.\ Th.\ Phys.\ {\bf 39}, 2717 (2000)], we formulate the cost in useful work of transforming one resource state into another as a linear program of convex optimization. This approach is based on the characterization of \emph{thermal quasiorder} given by Janzing \etal\ and later by Horodecki and Oppenheim [Nat.\ Comm.\ {\bf 4}, 2059 (2013)]. Both characterizations are related to an extended version of  majorization studied by Ruch, Schranner, and Seligman under the name \emph{mixing distance} [J.\ Chem.\ Phys.\ {\bf 69}, 386 (1978)].
\end{abstract}

\maketitle

\section{Introduction}
The recent advances in control of quantum and nano systems raise the question of the applicability of conventional thermodynamics in these new regimes. One promising approach to tackling this question is to regard thermodynamics as a \emph{resource theory} and then study this resource theory at the quantum level in order to determine which aspects of conventional thermodynamics persist in the new setting. 
From this point of view, the essence of thermodynamics is that not all transformations of physical systems are practically possible, and that this limitation gives rise to the notion of some physical systems being more \emph{useful} than others, in that they can be used to create the other states by the allowed operations. This is roughly the approach taken by Lieb and Yngvason to better understand the foundations of classical thermodynamics~\cite{lieb_physics_1999}.

In this paper we follow the related approach to thermodynamics as a resource theory in the quantum setting described by Janzing \etal\ \cite{janzing_thermodynamic_2000} and used by Brand\~ao \etal\ \cite{brandao_resource_2011} and Horodecki and Oppenheim~\cite{horodecki_fundamental_2011}. Here, the systems under consideration are explicitly treated in the framework of quantum mechanics, and transformations take the form of unitary operators. The resource theory specifies that only those transformations are allowed which commute with the Hamiltonians of the systems involved, and the only states which can be created at will are equilibrium Gibbs states at a fixed background inverse temperature $\beta$. The ultimate resource of the resource theory turns out to be useful work, free energy~\cite{brandao_resource_2011,horodecki_fundamental_2011}. Note that the resource theory generally applies to arbitrary systems, and is not restricted to, for instance, resources which are $n$-fold copies of a single-system state.

     Our contributions to the resource theory of thermodynamics are twofold. We first point out that conditions on the quasiorder of quasiclassical resources (resources in stationary states) described by Janzing \etal\ is in fact equivalent to the conditions found by Horodecki and Oppenheim, which they called \emph{thermomajorization}, and that both are manifestations of the \emph{mixing distance} of Ruch \etal~\cite{ruch_mixing_1978}. We then consider the question of the cost, in useful work, of transforming one quasiclassical resource state into another, and show that the quasiorder formulation of Janzing \etal\ provides a simple means to determine the work cost as a problem of convex optimization, a linear program. This problem was studied in a different setting by Egloff \etal~\cite{egloff_laws_2012}; an advantage of the present treatment is a significantly simpler proof. As a special case, our formulation recovers both the work value (or work cost) of a given resource state found by Horodecki and Oppenheim~\cite{horodecki_fundamental_2011} as well as the work cost of erasure, Landauer's principle~\cite{landauer_irreversibility_1961}.

\section{Thermal Quasiorder}
Any resource theory is defined by the allowed transformations and state preparations. In the thermodynamic setting, the allowed thermal operations are any energy-preserving, unitary actions on systems, plus the creation of Gibbs states at a fixed (inverse) temperature $\beta$, for any desired Hamiltonian~\cite{janzing_thermodynamic_2000,brandao_resource_2011}. 
Resources in this theory will be denoted $R=(\rho,\gamma)$ where $\rho$ is the state of the resource system, while $\gamma$ is the Gibbs state at temperature $\beta$ of the resource system. The unitary action is meant to describe any procedure that could in principle be performed, including those which call for manipulating the energy levels of the system by external fields or the use of interaction Hamiltonians forth; Ref.~\cite{brandao_resource_2011} describes more explicitly how these can be incorporated into the unitary model. 

Thermal operations generate a quasiorder of resource states: If a resource $R$ can be transformed into some other state $\tilde{R}$ by means of thermal operations, then we write $R\succ \tilde{R}$. The Gibbs state itself is the ``lowest'' state in the quasiorder. In particular, the thermal operations defined above are those given in Definition 7 of Janzing~\etal~\cite{janzing_thermodynamic_2000}, which envisions energy preserving transformations on three systems, the first in the state $\rho$ of the input resource, the second a heat bath, and the third the target system in its Gibbs state. Then, $R\succ\tilde{R}$ if there exists a $U^{ABC}$ such that 
\begin{align}
\tr_{AB}[U^{ABC}(\rho^A\otimes\hat{\gamma}^{B}\otimes\tilde{\gamma}^C)U^{\dagger ABC}]=\tilde{\rho}.
\end{align}

Observe that 
we do not attempt to transform $\rho$ ``directly'' into $\tilde{\rho}$, i.e.\ in the same state space. Instead, we use the heat bath to effect the transformation $\rho\otimes \tilde{\gamma}\rightarrow \eta\otimes \tilde{\rho}$, where the exhaust state $\eta$ is arbitrary. This accounts for differences in the overall zero of energy between two Hamiltonians: Given resource $R$ with Hamiltonian $H$, we can create $R'$ with Hamiltonian $H'=H+c$ by the thermal operation 
which simply swaps $A$ and $C$. 

In the case of quasiclassical resources, those which commute with the Hamiltonian and are therefore stationary states, Janzing \etal\ give the following complete characterization of the quasiorder. Only the eigenvalues of stationary states are relevant, so in this context we write $R=(p,g)$ with $p$ the eigenvalues and $g$ the Gibbs state probabilities, both interpreted as column vectors. With $e_n$ the length-$n$ column vector of $1$s, they show
\begin{theorem}[{\cite[Theorem 5]{janzing_thermodynamic_2000}}]
\label{thm:quasiorder}
Consider two quasiclassical resource states $R=(p,g)$ and $R'=(p',g')$, with dimensions $n$ and $n'$, respectively. Then $R\succ R'$ if and only if there exists an $n'\times n$ matrix $G$ such that 
\begin{enumerate}
\item $Gp=p'$
\item $Gg=g'$
\item $e_{n'}^{\rm T}G=e_n^{\rm T}$.
\end{enumerate}
\end{theorem}
\noindent The third condition fixes $G$ to be a stochastic matrix, i.e.\ one whose column sums are all unity. We shall call such stochastic matrices which preserve the Gibbs state \emph{Gibbs-stochastic}. 

Horodecki and Oppenheim~\cite{horodecki_fundamental_2011} formulate a similar, and as we shall see, equivalent result, which they term \emph{thermomajorization} due to its close connection with usual majorization. Indeed, the formulation of Janzing \etal\ is a generalization of the notion of \emph{$d$-majorization} by Veinott~\cite{veinott_least_1971} and is an instance of the \emph{mixing distance} of Ruch \etal~\cite{ruch_principle_1976,ruch_mixing_1978}. Marshall \etal~\cite{marshall_inequalities:_2009} provides a nice overview of known results involving $d$-majorization.

An important question regarding the thermal quasiorder is to find functions which preserve the order, called thermal monotones. One class is given by the $f$-divergences~\cite{csiszar_informationstheoretische_1963,morimoto_markov_1963,ali_general_1966}:
\begin{prop}
\label{thm:monotone}
All functions $\phi$ of the following form, with convex $f$, preserve the thermal quasiorder:
\begin{align}
\label{eq:monotone}
\phi(R)=\sum_i g_i \,f\!\left(\frac{p_i}{g_i}\right).
\end{align}
\end{prop}
\begin{proof}
The proof is a simple variation of an argument employed by Ruch and Mead~\cite[Theorem 1]{ruch_principle_1976}, which we omit here.
	 \end{proof}

Well-known examples of such thermal monotone functions are the relative entropies $D(p||g)=\sum_i p_i\log \frac{p_i}{g_i}$ and $D(g||p)$, which stem from $f(x)=x \log x$ and $f(x)=-\log x$, respectively, as well as the Renyi divergences $D_\alpha(p||g)=\frac1{1-\alpha}\log\sum_i p_i^\alpha g_i^{1-\alpha}$ with $\alpha\geq 0$, which follow from $f(x)=(x^\alpha-1)/(\alpha-1)$. 

Importantly, suitable subclasses of convex functions actually characterize the thermal quasiorder, as formalized in the following theorem by Ruch, Schranner, and Seligman,
\begin{theorem}[\cite{ruch_mixing_1978}]
\label{thm:altdef}
For resources $R$ and $R'$ let $r_i=p_i/g_i$ and $r'_i=p'_i/g'_i$.  Then the following are equivalent:
\begin{enumerate}
\item[\emph{(a)}]There exists a Gibbs-stochastic $G$ such that $Gp=p'$.
\item[\emph{(b)}]$\phi(R')\leq \phi(R)$ for all 
functions of the form \emph{(\ref{eq:monotone})}, with $f$ a \emph{continuous}, convex function,
\item[\emph{(c)}] $\int_0^t {\rm d}u\,\, r_{g}^{\prime *}(u)\leq \int_0^t{\rm d}u\,\, r^*_{g}(u)\, $  for all $0\leq t\leq 1$,
\item[\emph{(d)}]$\sum_i g'_i\left(r'_i-t\right)_+\leq \sum_i {g}_i\left(r_i-t\right)_+\,$ for all $t\in\mathbb{R}$,
\item[\emph{(e)}]$\sum_i g'_i|r'_i-t|\leq \sum_i g_i|r_i-t|\,$ for all $t\in\mathbb{R}$.
\end{enumerate}
Here $(a)_+=\max\{a,0\}$ and $r^*_g(u)$ denotes the decreasing rearrangement of $r$ by $g$: $r^*_g(u)=\sup\{s:m_r(s)>u\}$ for $0\leq u\leq 1$, with $m_r(s)=\sum_{i:r_i>s}g_i$, $s\geq 0$. 
\end{theorem}
Ruch, Schranner, and Seligman have established this statement in the more general setting of probability densities on the interval $[0,1]$. Condition (a) corresponds to definition (3f) in~\cite{ruch_mixing_1978}, (b) to (2a), (c) to (2e), (d) to (3c), and (e) to (3b). 
For the statement of (c) in the present discrete setting, we have however borrowed the more compact formulation due to Joe~\cite{joe_majorization_1990}. The integral in (c) defines the \emph{Lorenz curve} $L_R(t)$ for relative majorization~\cite{marshall_inequalities:_2009}. 
As with usual majorization, the Lorenz curve characterizes the conversion order in a simple geometric way, as shown in Figure~\ref{fig:lorenz}.

In fact, this is the same as the curve defined by Horodecki and Oppenheim~\cite{horodecki_fundamental_2011}, which is particular to the discrete setting and uses a different normalization. Their version has a much simpler definition, however, which is as follows (here we change the normalization). First, let $\pi$ be the permutation of indices of probabilities so that the sequence $(p_{\pi(i)}/g_{\pi(i)})_i$ is strictly non-increasing. Then the Lorenz curve is the piecewise linear function which joins the points given by the partial sums of $p_{\pi(i)}$ and $g_{\pi(i)}$~\cite{marshall_inequalities:_2009}, i.e.\ the points
\begin{align}
\left(t_k,L_R(t_k)\right)=\left(\sum_{i=1}^kg_{\pi(i)}\,,\,\sum_{i=1}^kp_{\pi(i)}\right).
\label{eq:dLcurve}
\end{align}

\begin{figure}[ht!]
\begin{tikzpicture}
    \begin{axis}[xmin=0,xmax=1,ymin=0,ymax=1,no markers,
    xtick={0,1},ytick={0,1},xlabel=$t$,ylabel=$L(t)$,
    label shift=-10pt,legend style={at={(.9,.1)},anchor=south east}]
    \addplot[draw=green!75!black,thick] table {L3.dat};
  \addlegendentry{$R_1$}
  \addplot[draw=red,thick] table {L1.dat};
  \addlegendentry{$R_2$}
   \addplot[draw=blue,thick] table {L2.dat};
 \addlegendentry{$R_3$}
 \addplot[draw=black,thick] table {L0.dat};
 \addlegendentry{$g$}
    \end{axis}
\end{tikzpicture}
\caption{\label{fig:lorenz} Lorenz curves of three resources $R_1$, $R_2$, $R_3$, and the Gibbs state $g$. A resource $R$ can be transformed into $\widetilde{R}$ if and only if the Lorenz curve of the former lies above that of the latter;  the Gibbs state has a flat Lorenz curve running from the origin to $(1,1)$. Here $R_1\succ R_3$ and $R_2\succ R_3$, but $R_1$ and $R_2$ are incomparable.}
\end{figure}
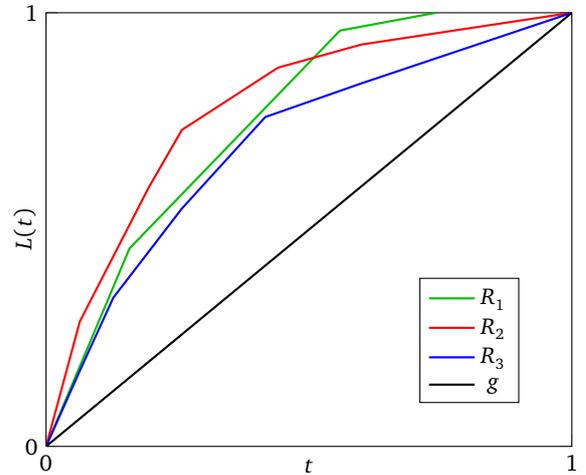

For a two-level system with energy gap $E$, we can easily work out the Lorenz curve explicitly. This is illustrative in its own right and will be useful later. The Gibbs state is described by $g=(\nicefrac1{1+e^{-\beta E}},\nicefrac1{1+e^{\beta E}})$, or equivalently $(Z_E(\beta)^{-1},Z_E(-\beta)^{-1})$, where $Z_E(\beta)={1+e^{-\beta E}}$. As there are just two levels, any quasiclassical state can be thought of as a Gibbs state at some temperature $\beta'$, so $p$ has the same form: $p=(\nicefrac1{1+e^{-\beta' E}},\nicefrac1{1+e^{\beta' E}})=(Z_E(\beta')^{-1},Z_E(-\beta')^{-1})$.

 Using \eqref{eq:dLcurve}, we need only give the single point at which the curve changes slope. To deal with the permutation $\pi$, we distinguish the two cases $\beta'>\beta$ and $\beta'<\beta$. In the former case, the resource state is colder than the   background Gibbs state; in the latter the resource is warmer, including situations in which $\beta'<0$ and there is a population inversion. When the resource is colder, no permutation in \eqref{eq:dLcurve} is needed, while the other case requires interchanging the two levels. One immediately finds that, for $\beta'>\beta$, the kink in the Lorenz curve occurs at the point $(t,L_R(t))=(Z_E(\beta)^{-1},Z_E(\beta')^{-1})$. For $\beta'<\beta$, the effect of interchanging the levels is just to take $\beta\rightarrow -\beta$  and $\beta'\rightarrow -\beta'$ in the previous analysis. The kink in the Lorenz curve is then at the point $(t,L_R(t))=(Z_E(-\beta)^{-1},Z_E(-\beta')^{-1})$. Figure~\ref{fig:lc1} shows curves for resources in the various regions.

\begin{figure}[h]
\begin{tikzpicture}
    \begin{axis}[
    xmin=0,xmax=1,ymin=0,ymax=1,no markers,
    xtick={0,.27,.73,1},ytick={0,1},
    extra y ticks={.85},
    extra y tick labels={$Z_E(-|\beta'|)^{-1}$},
    extra y tick style={tick label style={rotate=90}},
    xticklabels={0,$Z_E(-\beta)^{-1}$,$Z_E(\beta)^{-1}$,1},
    xlabel=$t$,ylabel=$L(t)$,
    label shift=-10pt,legend style={at={(.9,.1)},anchor=south east},
    grid=major
   ]

\filldraw[green!50!black,fill=green!50!black,draw opacity=0,fill opacity=0.25] (axis cs:0,.5) -- (axis cs:0,1) -- (axis cs:.5,1) -- (axis cs:.5,.5) -- cycle;  
  \draw[blue,draw opacity=0,fill=blue,fill opacity=0.25] (axis cs:.5,.5) -- (axis cs:0.5,1) -- (axis cs:1,1) -- cycle;
  \draw[red,draw opacity=0,fill=red,fill opacity=0.25] (axis cs:0,0) -- (axis cs:0,0.5) -- (axis cs:.5,.5) -- cycle;

  \addplot[draw=blue,thick] coordinates {
  (0,0)
  (.73,.85)
  (1,1)
  }; 
  \addlegendentry{$R_1$}
  
  \addplot[draw=red,thick] coordinates {
  (0,0)
  (.27,.45)
  (1,1)
  }; 
  \addlegendentry{$R_2$}
  
  \addplot[draw=green!50!black,thick] coordinates {
  (0,0)
  (.27,.85)
  (1,1)
  }; 
  \addlegendentry{$R_3$}
  
  \addplot[draw=black,thick] coordinates {
    (0,0)
    (1,1)
    };
  \addlegendentry{$g$}
  
  \end{axis}
\end{tikzpicture}
\caption{\label{fig:lc1} Lorenz curves of quasiclassical two-level resource states, whose Hamiltonian has an energy gap $E$, at background temperature $\beta>0$. Since there are just two states, any such resource state can be thought of as a Gibbs state at some temperature $\beta'$. $R_1$ denotes a state with $\beta'_1>\beta$, i.e.\ a colder system; the kink in the Lorenz curve of all resources of this form lies in the blue region for arbitrary $E>0$. $R_2$ has $\beta'_2< \beta$, which is hotter than the reference temperature; kink points of such resources fall in the red region. $R_3$ has the same inverse temperature as $R_1$, but negative, i.e.\ a state with population inversion; all resources with this property land in the green region. 
A resource in the excited state has $\beta'=-\infty$, while the state of an erased bit can be understood as the case $\beta'=\pm\infty$ and $\beta=0$.}
\end{figure}
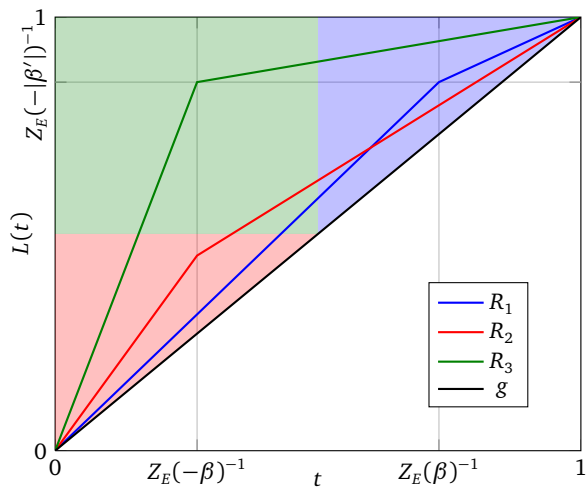
\vspace{-1cm}

\section{Work Cost of Transformations}
Given two resources $R$ and $R'$, suppose that it is not possible to transform $R$ into $R'$ using allowed thermal operations. Nonetheless, we expect that providing a sufficient amount of additional resources can make this transformation possible. Conversely, the transformation $R\rightarrow R'$ may be possible even if we additionally extract additional resources during the process. Traditionally, \emph{work} is standard resource in thermodynamics, often modelled as the change in the height of a weight. 

Here, we model the weight by an additional two-level system with energy gap $E$ in its excited state, and we denote this resource $A_E$. Then the work gain $W_{\rm gain}^\beta(R\to R')$ of the transformation can be defined as the largest $W$ such that 
\begin{align}A_{E}+R\succ A_{E+W}+R'\end{align} for some choice of $E>0$. %
That is, the transformation should produce the desired output $R'$ while increasing the gap of the additional system by $W$ and not producing any correlations between the two systems. If $W<0$, this represents the work cost required to drive the transformation.

It turns out that we may formulate a bound on the work cost or gain of implementing a the desired transformation in terms of a simple convex optimization, a linear program~\cite{barvinok_course_2002,boyd_convex_2004}. 
This approach is related to the results of Faist \etal~\cite{faist_quantitative_2012}, who studied the work cost of transformations between resources with completely degenerate Hamiltonians, but where preserving correlations with the environment are important. Closer to the present setting, Egloff~\etal~\cite{egloff_laws_2012} give an expression for the work cost which is related to the mixing distance of Ruch~\etal, but formulated in a somewhat different model of allowed operations than the set of thermal operations used here and which has a significantly more complicated proof. 

Before stating the result, let us first point out that while the question of whether the transformation $R\to R'$ is possible with thermal operations can be immediately formulated as a linear program, it is not so apparent that this holds for the work gain itself. To decide the former question, note that the three constraints of Theorem~\ref{thm:quasiorder} are linear in the entries of $G$, which must be themselves positive. Then the linear program which seeks to maximize $f(G)=0$ will find a feasible $G$ or certify that one does not exist. Specifically, if the optimal value of the dual problem turns out be unbounded, then there is no feasible $G$ (see, for instance, \cite[Theorem 8.2]{barvinok_course_2002}).

In a similar vein, we may formulate the task of finding $W_{\rm gain}^\beta(R\to R')$ as follows. First 
define $y=e^{-\beta W}$; we also drop the $\beta$ dependence in the partition function $Z_E$ since now its value is fixed. 
Then, for $G\in M_{n',n}(\mathbb{R})$, the set of real-valued $n'\times n$ matrices, $W_{\rm gain}^\beta(R\to R')=-\frac 1\beta \log y^*(R,R')$ in the optimization 
\begin{align}
\label{eq:nonlinprog}
\begin{array}{rrcl}
\textrm{find} 		&  y^*(R,R')				&=		&\min y\\[1mm]
\textrm{subject to} 	&  G(0,1)\otimes p		&=		&(0,1)\otimes p'\\
&G(1,e^{-\beta E})\otimes g&=&\frac{Z_E}{Z_{E+W}}(1,ye^{-\beta E})\otimes g'\\
&e_{2n'}^{\rm T}G 		&=& e_{2n}^{\rm T},\\
&y,E,G&\geq& 0,
\end{array}
\end{align}
where $G\geq 0$ is understood to mean that all components of $G$ are positive. Though the objective function is linear as before,  the constraints no longer are. 

Our main result is that the above can be transformed into a linear program valid in the limit $E\rightarrow \infty$:
\begin{theorem}
\label{thm:main}
Using thermal operations at inverse temperature $\beta$, a resource $R$ can be transformed into $R'$ in such a way that 
extracts an amount of work 
\begin{align}
W_{\rm gain}^{\beta} (R\rightarrow R')= -\tfrac1\beta\log x^*(R,R'),
\end{align}
for $x^*(R,R')$ the solution to the following linear program in the variables $x\in\mathbb{R}$ and $F\in M_{n',n}(\mathbb{R})$: 
\begin{align}
\label{eq:linprog}
\begin{array}{rrcl}
\textrm{find} &  x^*(R,R')&=&\min x\\[1mm]
\textrm{subject to} &  F p &= & p',\\
&F g &\leq & x g',\\
&e_{n'}^{\rm T}F&\leq &e_{n}^{\rm T},\\
&x,F&\geq& 0.
\end{array}
\end{align}
\end{theorem}
\begin{proof}
	The proof proceeds by showing showing that the solution to \eqref{eq:nonlinprog} is both less than and greater than the solution to \eqref{eq:linprog}. We begin with the case $x^*(R,R')\leq y^*(R,R')$.
	
	Suppose we have a feasible $y$, $E$, and $G$ in \eqref{eq:nonlinprog}. Every $E'\geq E$ would also lead to a feasible $y$ and $G$, since the resource $A_{E'}$ can be transformed to $A_E$ by thermal operations, as can be inferred from their Lorenz curves described in Figure~\ref{fig:lc1}. This will allow us to consider the limit $E\rightarrow \infty$ in what follows. 
	
	Any feasible $G$ can be written in block form as 
\begin{align}G=
\begin{pmatrix}
G_{11} & G_{12}\\
G_{21} & G_{22}
\end{pmatrix}.
\end{align}
Writing out the constraints in \eqref{eq:nonlinprog} in terms of the block decomposition, we obtain the follwing three pairs of equations. The constraints involving the resource are
\begin{align}
	G_{21}p&=0 \quad \text{and}\label{eq:pconsta}\\ 
	 G_{22}p&=p'.\label{eq:pconstb}
\end{align}
The constrains involving the Gibbs state read
\begin{align}
	G_{11}g+e^{-\beta E}G_{12}g&=\tfrac{Z_E}{Z_{E+W}}g'\quad\text{and}\label{eq:gconsta}\\
	G_{21}g+e^{-\beta E}G_{22}g &=xe^{-\beta E}g',\label{eq:gconstb}
	\end{align}
	where $x=y\frac{Z_E}{Z_{E+W}}$. 
Finally, normalization requires
\begin{align}
	e_{n'}^{\rm T}G_{11}+e_{n'}^{\rm T}G_{21}&=e_n^{\rm T}\label{eq:nconsta}\quad\text{and}\\
	e_{n'}^{\rm T}G_{12}+e_{n'}^{\rm T}G_{22}&=e_n^{\rm T}\label{eq:nconstb}.
\end{align}

Since 
$e_{n'}^TG_{12}$, $G_{21}g$, and $e^{-\beta E}$ are positive, the constraints that involve $G_{22}$ (the latter in each pair) 
immediately imply those of \eqref{eq:linprog}, with $F=G_{22}$. Therefore, any feasible $y$, $E$, and $G$ leads to a feasible $x$ and $F$. 

The remaining question is how the value of $x$ is related to that of $y$, and what this implies about the relation between $x^*(R,R')$ and $y^*(R,R')$.
There are two cases to consider. If $y\geq 1$, i.e.\ when $W\leq 0$, it holds that $x\leq y$. Thus it immediately follows that $x^*(R,R')\leq y^*(R,R')$. 
On the other hand, for $y\leq 1$ ($W\geq 0$), $x\geq y$. Now we consider the large $E$ limit. For large enough $E$ it holds that $x\leq y(1+e^{-\beta E}(1-\frac12y)-\frac12ye^{-2\beta E})$. So again we can infer that $x^*(R,R')\leq y^*(R,R')$ in the limit $E\rightarrow \infty$.

To show that $x^*(R,R')\geq y^*(R,R')$ we will first construct a feasible combination of $y$, $E$, and $G$ for \eqref{eq:nonlinprog} from a feasible choice of $x$ and $F$ in \eqref{eq:linprog}. 
First, set $G_{22}=F$ to satsify \eqref{eq:pconstb}. Then define $v=x g'-F {g}$, for which $v\geq 0$ by design, and set $G_{21}=e^{-\beta E} ve_n^{\rm T}$ for some $E$ to be specified later. This choice satisfies \eqref {eq:gconstb} with $x=e^{-\beta W}\frac{Z_E}{Z_{E+W}}$, and we have fixed the bottom row of $G$. 

For the top row, define $u^{\rm T}=e_n^{\rm T}-e_{n'}^{\rm T}F$, which is also positive by construction. 
Since both $p$ and $p'$ are normalized and $Fp=p'$, $u^{\rm T}p=0$. Therefore, by setting
$G_{12}=g'u^{\rm T}$, both \eqref{eq:pconsta} and \eqref{eq:nconstb} are satisfied. 

Two constraints remain to be satisfied, both involving $G_{11}$. Setting $G_{11}=tg'e_{n}^T$ for some $t$ to be chosen later, the two constraints now simplify to
\begin{align}
	t+e^{-\beta E}u^{\rm T}g&=\tfrac{Z_E}{Z_{E+W}}\quad\text{and}\\
	t+e^{-\beta E}e_{n'}^{\rm T}v&=1.\end{align}
Let us first confirm that the two are consistent and so our choice of $G$ is valid. Subtracting the expressions on the lefthand side yields
\begin{align}
	e^{-\beta E}(u^{\rm T}g-e_{n'}^{\rm T}v)&=e^{-\beta E}(1-x)\\
	&=e^{-\beta E}(1-y\tfrac{Z_E}{Z_{E+W}})\\
	&=e^{-\beta E}\frac{Z_{E+W}-e^{-\beta W}Z_E}{Z_{E+W}}\\
	&=e^{-\beta E}\frac{1-e^{-\beta W}}{Z_{E+W}}\\
	&=\frac{Z_E-Z_{E+W}}{Z_{E+W}},
\end{align}
which is indeed the righthand side. 
Finally, we must choose a value of $E$ such that both constraints are satisfied for positive $t$; this is always possible since both $u^{\rm T}g$ and $e_{n'}^{\rm T}v$ are bounded. Note that if some value $E$ ensures $t\geq 0$, then any $E'\geq E$ does as well. 

We have shown that a feasible $x$, $F$ implies the existence of a feasible $y$, $E$, and $G$. As before, we must now investigate the implications for the value of the objective function. Writing $y$ in terms of $x$ and $E$ we have 
\begin{align}
y=\frac x{1+e^{-\beta E}(1-x)}
\end{align}
If $x\leq 1$, then $y\leq x$ and we can immediately infer $y^*(R,R')\leq x^*(R,R')$. If $x>1$, we again consider large enough $E$, for which $y\leq x(1-\frac12 e^{-\beta E}(1-x))$. In the limit $E\rightarrow \infty$, we then recover $y\leq x$ and therefore $y^*(R,R')\leq x^*(R,R')$. 
\end{proof}

\section{Work Value of Resources \& Landauer's Principle}
Using the above linear program we can recover the work value or work cost of a given resource $R$ found in~\cite{horodecki_fundamental_2011}, the amount of useful work that can be obtained from $R$ or the amount required to create $R$. They additionally study the approximate work cost and gain, but here we deal only with the exact case. In the case of the work value, we are interested in $W_{\rm gain}^{\beta}(R,R')$ with $R'$ trivial. Thus, $p'=g'=e_1$, so the first condition is $Fp=1$. The third and fourth constraints fix $0\leq F_i\leq 1$. Since $p$ is a probability distribution, $F_i=1$ for all $i$ where $p_i\neq 0$. Now we must satisfy $Fg\leq x$. The smallest feasible $x$ can be obtained by setting $F_i=0$ for all $i$ where $p_i=0$. The optimum is $x^*(R,R')=\sum_{i:p_i\neq 0}g_i$, giving
\begin{align}
	W_{\rm gain}^\beta(R)=-\frac1\beta\log \sum_{i:p_i\neq 0}g_i.
\end{align} 
This is equation 4 of~\cite{horodecki_fundamental_2011}, for the case $\epsilon=0$ (the exact work value). 

The work cost of preparing $R$, meanwhile, is simply $-W_{\rm gain}^{\beta} (R',R)$ with $R'$ trivial. 
Now the first condition is simply $F=p$, while the third is automatically satisfied. The second condition becomes $p\leq xg$, so $x\geq p_i/g_i$ for all $i$. Therefore $x^*(R',R)=\max_i p_i/g_i$, giving
\begin{align}
	W_{\rm cost}^\beta(R)=\frac1\beta\log\max_i \frac{p_i}{g_i}.
\end{align}
This is equation 8 of~\cite{horodecki_fundamental_2011}, again in the $\eps=0$ case. 

We also immediately recover Landauer's principle~\cite{landauer_irreversibility_1961,janzing_thermodynamic_2000}. Here the goal is to transform an arbitrary two-level resource $R$ having a trivial Hamiltonian to the state $(1,0)$; one can easily extend the approach to an arbitrary number of levels. The linear program in \eqref{eq:linprog} has constraints $Fp=(1,0)$ for all $p$, as well as $Fe_2\leq xe_2$ and $e_2^{\rm T}F\leq e_2^{\rm T}$. As the first has to hold for any $p$, it follows that 
\begin{align}
F=\begin{pmatrix}1 &1\\0 &0\end{pmatrix},
\end{align}
and therefore $x\leq 2$. This gives a work cost of the transformation of $W_{\rm erase}^\beta(R)=\frac1\beta \log 2$, as expected.

\section{Conclusions}
We have shown that the thermal quasiorder of resources in the resource theory of thermodynamics is closely related to the notion of $d$-majorization, and we have given a characterization of the work cost or gain of operations on resource states in the resource theory of thermodynamics. Here we have adopted a definition of work in which an amount of work $W$ is gained when the energy gap of a two-level system in its excited state is increased by an amount $W$. This is not the only reasonable choice; Horodecki and Oppenheim consider transforming a two-level system with gap $W$ from its ground to its excited state~\cite{horodecki_fundamental_2011}, while Faist \etal\ measure work in terms of erased bits~\cite{faist_quantitative_2012}. Nonetheless, following the proof of Theorem~\ref{thm:main} with these different definitions of work gain leads back to the same result. 

The analysis of the work gain presented here proceeds under the assumption that the transformation is perfect, and determines the guaranteed amount of work available. It would be useful to try to formulate a simple convex optimization for the amount of work which can be gained by implementing the desired transformation, but which is guaranteed only with a probability greater than $1-\eps$ for some given $\eps$. 
Faist \etal\ have found such a convex optimization for trivial Hamiltonians~\cite{faist_quantitative_2012}. While the result of Egloff \etal\ includes an $\eps$-dependence~\cite{egloff_laws_2012}, they do not formulate it as a convex optimization, and the complexity of computing their expression in any given instance is unclear. 

\section{acknowledgments}
I thank Johan \AA berg, Fr\'ed\'eric Dupuis, Philippe Faist, Renato Renner, 
Paul Skrzypczyk, Michael Walter, and Nicole Yunger Halpern for very helpful discussions. This work was supported 
by the Swiss National Science Foundation (through the National
Centre of Competence in Research `Quantum Science and Technology' and grant
No. 200020-135048) and by the European Research Council (grant 258932).

\bibliographystyle{apsrev4-1}
\bibliography{hypothermo}
\end{document}